\documentclass[11pt]{article}
\usepackage{geometry}
\usepackage{amssymb}
\usepackage{mathrsfs}
\usepackage{amsmath}
\usepackage{times}

\newtheorem{theorem}{Theorem}[section]
\newtheorem{corollary}[theorem]{Corollary}

\newtheorem{proposition}[theorem]{Proposition}
\newtheorem{claim}[theorem]{Claim}

\newtheorem{observation}[theorem]{Observation}

\def\squarebox#1{\hbox to #1{\hfill\vbox to #1{\vfill}}}
\newcommand{\qed}{\hspace*{\fill}
\vbox{\hrule\hbox{\vrule\squarebox{.667em}\vrule}\hrule}\smallskip}

\newenvironment{proof}{\noindent{\bf Proof:~~}}{\(\qed\)}

\newcommand{\vals}{s}
\newcommand{\valb}{b}
\newcommand{\ffs}{f_s}
\newcommand{\fb}{f_b}
\newcommand{\Fs}{F_s}
\newcommand{\Fb}{F_b}
\newcommand{\lowerint}{\underline{a}}
\newcommand{\upperint}{\overline{a}}

\begin{document}
\title{(Almost) Efficient Mechanisms for Bilateral Trading}
\author{Liad Blumrosen\thanks{A preliminary version of some of the results in this paper was presented in ACM EC 2014 in a paper titled "Reallocation Mechanisms".} \and Shahar Dobzinski}
\maketitle
\begin{abstract}
We study the bilateral trade problem: one seller, one buyer and a single, indivisible item for sale. It is well known that there is no fully-efficient and incentive compatible mechanism for this problem that maintains a balanced budget. 
We design simple and robust mechanisms that obtain approximate efficiency with these properties. We show that even minimal use of statistical data can yield good approximation results. Finally, we demonstrate how a mechanism for this simple bilateral-trade problem can be used as a ``black-box'' for constructing mechanisms in more general environments.
\end{abstract}


%


%
%
%

\section{Introduction}

In this paper we revisit the simple and well studied model of bilateral trade 
\cite{MS83}.
In the bilateral trade problem, a single seller is the owner of one indivisible item which can be traded with a single buyer.
Both the seller and the buyer have privately known values for consuming the item; These values are unknown to the market designer who only has access to their probability distributions.

The bilateral trade problem is probably the simplest form of two sided markets and supply chains. Two-sided markets, where multiple strategic sellers and buyers compete on each side, have always been an important form of trade. The Internet era emphasized the importance of understanding such markets with the emergence of large-scale online platforms like eBay and Amazon, and, more recently, with the popularity of the sharing economy and two sided platforms like Uber and Airbnb. Our paper raises new insights regarding the simple bilateral-trade settings, and we will argue that this understanding could help in the design of more general markets either directly or indirectly.

We will have three requirements from a mechanism for the bilateral trade problem. The first one is (ex-post) \emph{individual rationality}: the participation of the agents is voluntary and at any point they may leave the market and consume their initial endowments. The second requirement is (ex-post) \emph{budget balance}: the mechanism is not allowed to subsidize the agents or to make any profits. The third requirement is \textit{incentive compatibility} (IC). While most of the literature focuses on Bayes-Nash incentive compatibility, our mechanisms will possess the stronger (and thus more desired) property of dominant-strategy incentive compatibility (DSIC).

The seminal paper by Myerson and Satterthwaite \cite{MS83} analyzes the bilateral trade problem. Their celebrated impossibility result states that even in this simple setting, there is no Bayes-Nash incentive compatible mechanism which is fully efficient, individually rational and budget balanced. Myerson and Satterswaite go on and characterize the ``second best" mechanism, i.e., the mechanism that maximizes the expected total gains from trade subject to the individual rationality, budget balance, and Bayes-Nash incentive constraints. Notice that since maximizing the expected total gains from trade is equivalent to maximizing the expected improvement in efficiency, Myerson and Satterthwaite characterize in fact the most efficient mechanism subject to individual rationality, budget balance and incentive constraints.

Unfortunately, this ``second best" mechanism is not given as a closed-form formula, and it might be a non-trivial task to explicitly describe it and compute the expected efficiency even for simple distributions. Our main goal in this paper is to develop simple and practical mechanisms for the bilateral trade problem that will approximately maximize the expected efficiency. En route, we will obtain good bounds on the efficiency of Myerson and Satterthwaite's ``second best'' mechanisms.

\subsection{The Main Results}

Our work is inspired by a recent line of research that studies the power of simple mechanisms in comparison to the performance of the optimal, but complex, mechanisms (see, e.g., \cite{BNS06,HR09,DRY10,mcA92,mcA08} and references within). We restrict our attention to mechanisms that are easier to understand and implement. Specifically, we focus on mechanisms that simply post a take-it-or-leave-it price; A trade occurs only if both agents accept this price. In such mechanisms, the agents have obvious dominant strategies, and in the equilibrium analysis we do not need to speculate whether the agents compute the equilibrium correctly or converge to a particular Bayes-Nash equilibrium.

The simplicity of our mechanisms comes at the obvious price of sub-optimality, as the efficiency of our mechanisms might be inferior to that of the second-best mechanism. However, we are able to quantify this loss and show that our simple mechanisms perform quite well even with respect to a stronger benchmark: the optimal ex-post efficiency (i.e., ``first-best'' efficiency or simply the expected value of the maximum between the seller's and the buyer's value for the item).
We measure the quality of our mechanisms by the fraction of the optimal efficiency that they guarantee\footnote{Another possible objective may be approximating the expected gains from trade (GFT). As we have noted, a mechanism that maximizes gain-from-trade is fully efficient as well. However, from an approximation point of view,
the two objectives behave differently: indeed, any $c$ approximation to the GFT is also at least a $c$ approximation to the efficiency. However, the converse is false and in particular in Appendix \ref{app:gft-imposs} we show that no DSIC, budget balance and individual rational mechanism can guarantee a constant approximation to the optimal GFT. On the brighter side, McAfee \cite{mcA08} provides a mechanism that achieves a $1/2$ approximation to the optimal GFT when the median of the seller's distribution is no greater than the median of the buyer, and more recent approximations are in \cite{BM16,CGKLT17,BCWZ17}.}. Our main result is as follows:



\vspace{0.1in} \noindent \textbf{Theorem 1:} For every pair of distributions, there is a take-it-or-or-leave-it  price mechanism that achieves at least $1-\frac{1}{e}$ of the optimal efficiency in dominant strategies. The mechanism is obviously individually rational and budget balanced.

\vspace{0.1in}\noindent We stress that this is a worst-case bound over all possible distributions; typically there is a price that achieves a much higher fraction of the optimal welfare.
A recent impossibility result by \cite{KV19} shows that no dominant-strategy mechanism which is also individually rational and budget balanced can guarantee more than $0.7385$ of the optimal welfare (improving over a previous hardness result of 0.7485 that was provided by \cite{CKLT16}). Together with our $1-\frac{1}{e} \simeq 0.632$ bound this leaves a relatively small gap.\footnote{
In an earlier version of this paper \cite{BD14} we proved the existence of a price that guarantees $0.51$ of the optimal welfare. The follow-up \cite{CKLT16} improved this bound to $0.52$, and here we improve it further to $1-\frac{1}{e}$. Following our improvements, Kang and Vondr{\'a}k \cite{KV19} have further improved the approximation ratio to $1-\frac 1 e+0.0001$ by a much more complicated mechanism. As far as we know, we were the first to approximate efficiency for the general bilateral-trade problem.
}
One immediate corollary of the above theorem is a bound on the performance of the most efficient (second-best) mechanism with Bayes-Nash incentive compatibility. That is, it shows that this mechanism (that was characterized by Myerson and Satterthwaite in \cite{MS83}) achieves at least $1-\frac{1}{e}$ of the optimal efficiency even for the worst pair of distributions.

Given two distributions, one can find the price that maximizes the expected revenue simply by considering all possible values in the support and calculating the expected revenue when the price equals this point. However, while one cannot ignore the wealth of distributional information that retailers have at their disposal, this knowledge is often incomplete.
In many scenarios, statistical knowledge is not available to the designers (for example, in new markets or markets that involve new participants) or it may be costly. In other cases, the designer can only accurately estimate some statistics of the distributions (expectation, median, etc.) rather than the full details of the distribution.

Our next result shows that the $1-1/e$ approximation can be achieved by a randomized mechanism that chooses the prices according to a carefully chosen distribution. Importantly, this improved approximation only needs access to the distribution of the seller.

\vspace{0.1in} \noindent \textbf{Theorem 2:} For every distribution of the seller, there is a distribution $G$ of prices with the following property: if we post a take-it-or-leave-it price sampled from G then for \emph{every} distribution of the buyer we obtain a fraction of at least $(1-1/e)\approx 0.63$ of the optimal efficiency. 

\vspace{0.1in} \noindent Our main result, stated as Theorem 1, is actually a direct corollary of Theorem 2. Theorem 2 says that if prices are chosen at random according to some distribution, then the expected efficiency is at least $1-1/e$ of the optimal efficiency. It follows that for every distribution of the buyer, there must be one price that attains at least $1-1/e$ of the optimal efficiency as well. To explicitly find this price one needs to know the details of the distribution of the buyer as well.

We also present two simple mechanisms that achieve at least $1/2$ of the full efficiency but use very restricted distributional knowledge.\footnote{
At first glance, achieving 1/2 of the optimal efficiency looks trivial: we can allocate the item to the player with the higher expected efficiency and that should give us at least 1/2 of the optimal efficiency. However, this overly-optimistic argument fails since in bilateral trade the two agents are not symmetric: while it is easy to leave the item with the seller, we need to convince the seller to relinquish the item (via a sufficiently high payment) for the two agents to trade.
}
The first mechanism posts a price that is equal to the median value of the seller's distribution, and the other mechanism computes some sort of a weighted median based only on the distribution of the buyer and offers this price to the agents.
We show that the $1/2$ bound is tight in the sense that no single take-it-or-leave-it price can achieve a better approximation knowing only the distribution of the seller or only the distribution of the buyer.


\subsection{Using Solutions for Bilateral Trade as ``Black-Boxes"}

The above results study the fundamental problem of Bilateral Trade, and insights from the study of this basic problem should alleviate the design of mechanisms in more complex trade environments.\footnote{
In a companion paper \cite{BD14}, we show how ideas that stem from our understanding of the simple bilateral-trade setting can be used in the design of mechanisms in more general, multi dimensional exchange settings. In particular, we show in \cite{BD14} how a careful use of distribution quantiles as prices can help approximate the optimal efficiency; In prior-free settings, quantile-like prices are artificially created by random sampling methods.
}
In the second part of the paper, we show how results for the bilateral trade problem can be directly applied to construct approximately-optimal mechanisms in other settings in a black box manner.
More specifically, we show reductions\footnote{
A \emph{reduction} is a fundamental and widely used technique in theoretical computer science. This is essentially an algorithm that transforms one problem into another problem such that any solution to the latter problem implies a solution to the original problem as well.
} of the following form: given a mechanism that guarantees some $\alpha$ approximation to Bilateral Trade, we design a mechanism for the other problem that obtains $f(\alpha)$ approximation for some function $f$.
We assume that we are given a DSIC, ex-post budget balanced and individually-rational mechanism as input, and the mechanism we output maintains these desired properties.
Given an $\alpha$-approximation mechanism for Bilateral Trade, we give the following reductions:
\begin{itemize}
\item \emph{Partnership Dissolving.} Consider $n$ players, each player $i$ initially owns a fraction $r_i$ of a divisible good. Players have linear valuations, such that the value for player $i$ for a fraction $c$ of the item is $c \cdot v_i$. This is the classic model by Cramton, Gibbons and Klemperer \cite{CGK87}. We design a mechanism that achieves an $\alpha$ fraction of the optimal efficiency.
    This result proves that the extreme-ownership scenario of the partnership dissolving problem is actually the hardest to solve, and any approximation for the first implies the same approximation factor for the latter problem.
\item \emph{Divisible good with general monotone valuations.} Consider a seller and a buyer, where the seller initially owns a fully divisible good. The players have monotonically non-decreasing valuations. We design a mechanism that guarantees $\frac{\alpha}{\alpha+1}$ of the optimal efficiency in this setting.
\item \emph{Divisible good with convex valuations.} Consider a 2-player environment, where each player $i$ initially holds some fraction $r_i$ of a divisible good. The players have convex valuations, i.e., preferences with decreasing marginal valuations. We design a mechanism that guarantees $\frac{\alpha}{\alpha+1}$ of the optimal efficiency in this setting.
\end{itemize}

Applying our $(1-1/e)$-approximation mechanism mentioned above, we obtain a $0.39$-approximation mechanisms for the two 2-player trade problems with a divisible good, and a $(1-1/e)$-approximation mechanism to the $n$-player partnership dissolving problem. Any future improvements to the $1-1/e$ bound would immediately imply an improvement to the above bounds as well.

\paragraph*{More related research.}
The paper of McAfee \cite{mcA08} is closest in spirit to ours. McAfee shows that for distributions such that the median of the buyer is greater than the median of the seller, posting any price between these two medians achieves at least half of the optimal expected gains from trade. In a sense we show that this median technique can be developed to approximate the social welfare for every pair of distributions.

A recent paper by Garratt and Pycia \cite{GP15} shows that with non quasi-linear preferences, e.g., with risk aversion and wealth effects, efficient trade is possible under some general conditions on the information structure.
The partnership dissolving model was first studied by Cramton, Gibbons and Klemperer \cite{CGK87} who show that if the shares are close enough to equal holdings, there exists a fully efficient Bayes-Nash incentive compatible mechanism (see also \cite{mcA92}, and a survey \cite{Mol01}).

Our paper is part of two notable research directions from the last decade.
The first is a series of paper that show how simple, feasible mechanisms can approximate the results achieved by the optimal, yet often too complex, mechanisms (e.g., \cite{BNS06,HR09,DRY10,CHMS10,BILW14}). The second line of research studies ``robust" mechanism design, which is an effort to design mechanisms that are less sensitive to the fine details of the specific environment that may not be available to the planner in practice (see the survey \cite{BM13} and the many references within).

\section{Model}

The bilateral trade problem involves two agents, a seller and a buyer.
The seller owns an indivisible item, and his valuation for consuming this item is $\vals$. If a trade occurs and the item is allocated to the buyer, then the buyer enjoys a value of $\valb$. 

$\vals$ and $\valb$ are independently distributed according to distributions $\Fs$ and $\Fb$, respectively. All our results hold for any pair of distributions, but for the simplicity of presentation we assume the existence of density functions $\ffs$ and $\fb$, for the seller and the buyer respectively, which are always positive on a support $[\lowerint,\upperint]$ ($\lowerint\geq 0$) and atomless.


For every realization of $\vals$ and $\valb$, the deterministic mechanism takes $\vals,\valb$ as input and determines the allocation $X(\vals,\valb)\in \{0,1\}$, the payment $p_s$ that the seller receives, and the payment $p_b$ of the buyer. We set $X(\vals,\valb)=0$ in the case of no trade (the seller keeps the item) and $X(\vals,\valb)=1$ if the buyer receives the good.

In this paper we restrict our attention to ex-post budget balanced mechanisms in which for all $\vals,\valb$ we have that $p_{\vals}=p_{\valb}$.\footnote{
This requirement is sometimes known as \emph{strong} budget balance. Weak budget balance allows the mechanism to accumulate profit. We are able to prove our positive results even with the strong budget balance property.
} We can therefore denote the trade price by $p$. We also consider only ex-post individually rational mechanisms, in which $p=0$ if there is no trade and $\vals\leq  p\leq \valb$ otherwise.

The private information of the buyer is $\valb$ and the private information of the seller is $\vals$. When trade occurs at price $p$, the seller's utility is $p$ and the buyer's utility is $\valb-p$. With no trade, the seller's utility is $\vals$ and the buyer's utility is $0$. 

Our mechanisms are dominant strategy incentive compatible (DSIC). In particular, we consider posted price mechanisms, in which the center posts a trade price $p$ that is independent on the realizations of the valuations and trade occurs only if both players agree to that price. Posted price mechanisms are clearly DSIC, budget balanced and individually rational. In fact, since\footnote{One can also relax our requirements and consider mechanisms with $X(\vals,\valb)\in [0,1]$ with all definitions extended accordingly. However, Hagerty and Rogerson \cite{HR87} essentially show that every optimal mechanism is a posted price mechanism. Other papers (e.g., \cite{SZ16,DK15}) show the optimality of posted price mechanisms in related settings.} $X(\vals,\valb)\in \{0,1\}$ it is not hard to see that every DSIC, budget balanced and individually rational mechanism is a posted price mechanism, simply because $p$ cannot depend on $\valb$ nor on $\vals$.

A fully efficient mechanism will initiate a trade whenever $\vals>\valb$. Define the optimal efficiency as
\begin{align*}
E_{ \vals \sim \Fs, \valb \sim \Fb }[\max\{ s,b\}]
\end{align*}
We measure the performance of our mechanisms by the fraction of the optimal efficiency that they obtain. Specifically, we denote the expected efficiency achieved by a mechanism with an allocation function $X$ as
\begin{align*}
E_{ \vals \sim \Fs, \valb \sim \Fb }[ (1-X(\vals,\valb))\cdot \vals + X(\vals,\valb) \cdot \valb]
\end{align*}



\section{Warm Up: Simple Mechanisms for Bilateral Trade}
\label{sec:warmup}

We begin our investigations by introducing two simple mechanisms that
guarantee at least half of the optimal efficiency.
In the next sections we provide an improved approximation ratio that relies on ideas that are presented here in their most crystallized form.

The mechanisms that are developed in this section -- as well as all of our mechanisms for the bilateral trade problem -- simply set a price $p$. The item is traded if and only if $\valb\geq p$ \emph{and} $\vals\leq p$. In case of trade, the buyer pays $p$ and the seller receives the same amount. Since the choice of $p$ will depend only on the the distributions of the players and not on the actual realizations, this family obviously yields dominant strategy mechanisms that are individually rational and budget balanced.

Our first mechanism is the \emph{median mechanism}\footnote{As discussed in the introduction, this mechanism was has shown to approximate the GFT in some settings by \cite{mcA08}.} sets the price $p$ to be the median of seller's distribution, i.e., the point $M_s$ such that $\Fb(M_s)=\frac 1 2$.

\begin{theorem}
The median mechanism is DSIC, individually rational, budget balanced and always achieves at least $1/2$ of the optimal social welfare.
\end{theorem}
\begin{proof}
By the discussion above, the mechanism is obviously DSIC, individually rational, and budget balanced. We now analyze its approximation ratio. We start with some notation. Since the trade price depends only on the distribution of the seller, to analyze the approximation ratio we may fix the value $\valb$ of the buyer and show that for every fixed $\valb$ the expected approximation ratio is $1/2$ in expectation over the seller's distribution. We divide our analysis to two disjoint cases. In both cases we use the observation that the expected optimal welfare is at most $\Pr[\valb\geq \vals]\cdot \valb+\Pr[\vals>\valb]\cdot E[\vals|\vals>\valb]$.

\begin{enumerate}
\item $\valb\geq M_s$. In this case the item is sold with probability $\frac 1 2$, when $\vals\leq M_s$. We get that the approximation ratio of the mechanism (i.e., the expected optimal efficiency divided by the expected efficiency of the median mechanism) is at most:
\begin{align*}
&\frac {\Pr[\valb> \vals]\cdot \valb+\Pr[\vals\geq\valb]\cdot E[\vals|\vals\geq \valb]} {\Pr[\vals<M_s] \cdot \valb+ \Pr[\vals\geq  M_s] \cdot E[\vals|\vals\geq M_s]} \\&=
 \frac {\Pr[\valb> \vals]\cdot \valb+\Pr[\vals\geq\valb]\cdot E[\vals|\vals\geq\valb]} {\Pr[\vals<M_s] \cdot  \valb+ \Pr[\vals\geq \valb] \cdot E[\vals|\vals\geq \valb] + \Pr[M_s< \vals< \valb] \cdot E[\vals|M_s< \vals < \valb]} \\&\leq
 \frac {\Pr[\valb\geq \vals]\cdot \valb+\Pr[\vals\geq\valb]\cdot E[\vals|\vals\geq \valb]} {\frac 1 2 \cdot  \valb+ \Pr[\vals\geq\valb]\cdot E[\vals|\vals\geq \valb] } \leq 2
\end{align*}

\item $\valb< M_s$. Here the item is never sold and the expected welfare is $E[\vals]$. We bound the expected approximation ratio:
\begin{align*}
\frac {\Pr[\valb\geq \vals]\cdot \valb+\Pr[\vals>\valb]\cdot E[\vals|\vals>\valb]} {E[\vals]} &\leq \frac {\frac 1 2 M_s+\Pr[\vals>\valb]\cdot E[\vals|\vals>\valb]} {\Pr[\vals\leq \valb]\cdot E[\vals|\vals\leq \valb]+\Pr[\vals>\valb]\cdot E[\vals|\vals>\valb]}\\
&\leq \frac {\frac 1 2 M_s+\Pr[\vals>\valb]\cdot E[\vals|\vals>\valb]} { \Pr[\vals>\valb]\cdot E[\vals|\vals>\valb]}\\
&\leq \frac {\frac 1 2 M_s+\frac 1 2 M_s} {\frac 1 2 M_s} = 2\\
\end{align*}
where in the last transition we use the fact that $\Pr[\vals>\valb]\cdot E[\vals|\vals>\valb]=\Pr[M_s\geq \vals>\valb]\cdot E[\vals|M_s\geq \vals>\valb]+\Pr[\vals>M_s]\cdot E[\vals|\vals>M_s]\geq \frac 1 2 M_s$.
\end{enumerate}
\end{proof}

One could hope that the ``symmetric'' mechanism that sets the trade price to be $M_b$ (the median of the distribution of the buyer) will provide a good approximation ratio as well. Unfortunately, this is not the case. To see this, consider the distribution where the buyer's value is $\epsilon$ with probability $\frac 1 2 + \epsilon$ and $t>1$ with probability $\frac 1 2 -\epsilon$. Let the seller's value be $1$ with probability $1$. When we set the trade price to be $M_b=1$ the item is never sold since the price is always too low for the seller. The welfare of the mechanism is $1$ while ideally we would like to sell the the item whenever the buyer's value is $t$ and get an expected approximation ratio of $(\frac 1 2 -\epsilon)t$. Letting $t$ approach infinity and $\epsilon$ approach $0$, we get that setting the trade price to be $M_b$ does not guarantee any constant approximation ratio.

However, we do show that a more careful choice of a trade price $p$ (using only the distribution of the buyer) results in an approximation ratio of $1/2$. The \emph{weighted median mechanism} sets the trade price $p$ to be the point $W_b$ for which $\Fb(W_b)\cdot E[\valb|\valb<W_b] = (1-\Fb(W_b))\cdot E[\valb|\valb\geq W_b]$. In particular, since $E[\valb|\valb \geq W_b] \geq E[\valb|\valb<W_b]$ we get that $\Pr[\valb>W_b]\leq \frac 1 2$. In addition:
\begin{align}\label{eqn-weighted-median}
(1-\Fb(W_b))\cdot E[\valb|\valb\geq W_b] =\Fb(W_b)\cdot E[\valb|\valb<W_b] = \frac {E[\valb]} 2
\end{align}

\begin{theorem}
The weighted median mechanism is DSIC, individually rational, budget balanced and always achieves at least  $1/2$ of the optimal social welfare.
\end{theorem}
\begin{proof}
Again, we only analyze the approximation ratio of the mechanism. Similarly to the analysis of the median mechanism, since the trade price $p$ depends only on the buyer's distribution we may analyze the approximation ratio assuming that $\vals$ is fixed. We divide our analysis into two cases:
\begin{enumerate}
\item $\vals<W_b$. In this case the item is sold whenever the buyer's value is at least $W_b$. The approximation ratio can be bounded by:
\begin{align*}
\frac {\Pr[\valb\geq \vals]\cdot E[\valb|\valb\geq \vals]+\Pr[\vals>\valb]\cdot \vals} {\Pr[\valb\geq W_b] \cdot E[\valb|\valb\geq W_b]+\Pr[\valb<W_b] \cdot \vals}
\leq \frac {E[\valb]+\vals} {\frac {E[\valb]} {2}+\frac {\vals} 2}\leq 2
\end{align*}
where we use Equation (\ref{eqn-weighted-median}) and $\Pr[\valb<W_b]\geq \frac 1 2$.

\item $\vals\geq W_b$. In this case the item is never sold. We will make use of $E[\valb]=\Pr[\valb>W_b]\cdot E[\valb|\valb>W_b]+\Pr[\valb\leq W_b]\cdot E[\valb|\valb\leq W_b]\leq \frac 1 2 \cdot E[\valb] + W_b$ which implies that $\frac {E[\valb]} 2\leq W_b\leq \vals$. Therefore:
\begin{align*}
\frac {\Pr[\valb\geq \vals]\cdot E[\valb|\valb\geq \vals]+\Pr[\vals>\valb]\cdot \vals} {\vals} &\leq \frac {\Pr[\valb\geq W_b]\cdot E[\valb|\valb\geq W_b]+\Pr[\vals>\valb]\cdot \vals} {\vals} \\ &\leq \frac {\frac {E[\valb]} 2+\vals} {\vals} \leq \frac {\vals+\vals} {\vals} = 2
\end{align*}
where we again use Equation (\ref{eqn-weighted-median}).
\end{enumerate}
\end{proof}

The two simple mechanisms that we presented in this section prove that one can achieve good approximation to the optimal efficiency while preserving budget balance and individual rationality. Moreover, this can be done having statistical knowledge only on one side of the market. 
It turns out that this is the best ratio that can be achieved via deterministic 
DSIC mechanisms (proof is in Appendix \ref{app:subsec:median-is-tight}).

\begin{proposition}\label{prop-bilateral-distribution-info}
No deterministic DSIC, individually rational and budget balanced mechanism can obtain an approximation ratio better than $1/2$ if the mechanism uses only the distribution of the seller or only the distribution of the buyer.
\end{proposition}


\section{A (1-1/e)-Approximation Mechanism}

In Section \ref{sec:warmup} we showed that one can compute a price that depends only on the distribution of the seller and this price will guarantee at least half of the optimal efficiency no matter what the distribution of the buyer is. We showed that this bound of $\frac{1}{2}$ is tight for all deterministic dominant-strategy incentive-compatible mechanisms. However, this bound could be improved by randomly choosing a price.
We will describe a mechanism that posts a price at random, and we will only require this mechanism to have access to the distribution of the seller. We will then conclude the existence of a deterministic posted-price mechanism that has access to both distributions.

\vspace{2mm}

We let $q(\cdot)$ denote the quantile function of the seller, that is, $\Fs(q(x))=x$ for every $x \in [\lowerint,\upperint]$.
Consider the following mechanism that posts a price to both players, but chooses the price randomly. We call this mechanism the  \emph{Random-Quantile mechanism}:
\begin{itemize}
\item \textit{Choose a number $x \in [1/e,1]$ according to the distribution with CDF $G(x)=\ln(e\cdot x)$.}
\item \textit{Set the price to be $q(x)$.}
\end{itemize}

\begin{theorem}
The Random-Quantile mechanism is DSIC, individually rational, budget balanced, and
achieves in expectation at least $1-\frac{1}{e}$ fraction of the optimal efficiency for every pair of distributions $\Fs$ and $\Fb$.
\end{theorem}

Before providing the proof of the theorem, we note that our Random-Quantile mechanism was shown by \cite{KV19} to provide the best approximation ratio among all ``quantile'' mechanisms: mechanisms that choose $x$ according to some distribution, and then set the price to be $q(x)$.

\begin{proof}
We prove the theorem for every fixed value $\valb$ of the buyer and it will clearly hold for every distribution of the buyer as well.

We will first prove that if we truncate the seller's distribution at $b$, the approximation ratio can only get worse.

\begin{claim}
Let the buyer have a fixed value $\valb$,
and let $\Fs^*$ be a distribution for the seller such that $\Fs^*(x)=\Fs(x)$ for $x<b$, $\Fs^*(x)=1$ for $x\geq b$.
We denote by $RAND(F_s,b)$  the expected efficiency of the Random-Quantile mechanism for the distribution $\Fs$ for the seller and a fixed value $\valb$ for the buyer. We denote by $OPT(\Fs,b)$ the optimal efficiency for these distributions.
Then,
\begin{align*}
\frac{RAND(\Fs,b)}{OPT(\Fs,b)} \geq \frac{RAND(\Fs^*,b)}{OPT(\Fs^*,b)}
\end{align*}
%


\end{claim}
\begin{proof}
We first note that
\begin{align}
RAND(\Fs,\valb)-RAND(\Fs^*,\valb) & =
(1-\Fs(\valb))\Big( E[\vals|\vals \geq \valb]-\valb \Big)\\
&= OPT(\Fs,\valb)-OPT(\Fs^*,\valb)
\end{align}
To see this, note that that for all values of $\vals$ that are below $b$, the expected efficiency is the same for $\Fs^*$ and $\Fs$.
For values $\vals$ that are above $b$, both the optimal allocation and the Random-Quantile mechanism gain efficiency of $\vals$ for $\Fs$ but
of only $\valb$ for $\Fs^*$.
Thus, the difference between the expected efficiency in the Random-Quantile mechanism and in the optimal efficiency is exactly $(1-\Fs(\valb))\Big( E[\vals|\vals \geq \valb]-\valb \Big)$.


Therefore, moving from $\frac{RAND(\Fs,b)}{OPT(\Fs,b)}$ to
$\frac{RAND(\Fs^*,b)}{OPT(\Fs^*,b)}$ the enumerator and the denominator decrease by the same additive term. The claim follows as the two fractions are at most $1$.
%
\end{proof}

We now analyze the approximation ratio assuming a fixed buyer valuation $b$ and the seller distribution $\Fs^*$, for which $OPT=b$. This will also imply the same approximation bound for the original distributions by the above claim. Let 
$y=\Fs(b)$ denote the probability that the seller's value is below $\valb$ for $\Fs^*$.

In our randomized mechanism, either the seller keeps the item or the buyer gets it. We will separately bound the expected contribution of the seller's value and the buyer's value in the mechanism and by the linearity of expectation get a bound on the overall efficiency.
Recall that the seller accepts a price $q(x)$ with probability $x$ by definition and that the density of the price distribution $G(x)$ is $\frac{1}{x}$.
When the mechanism chooses a quantile $x$ between $\frac{1}{e}$ and $y$, a trade occurs with probability $x$, and the realized efficiency is $\valb$.
Therefore, the contribution of the buyer to the overall efficiency is at least $\int_{\frac{1}{e}}^{y} x \cdot b \cdot \frac{1}{x} dx$.
On the other hand, note that unless the seller's value is strictly below $b$, the efficiency is $\valb$ with probability $1$.
The seller has value $b$ with probability $1-y$, and therefore the contribution of the seller's value to the overall efficiency is at least $b\cdot (1-y)$.
We have that

%


\begin{align}
RAND(\Fs^*,b) &\geq \int_{\frac{1}{e}}^{y} x \cdot b \cdot \frac{1}{x} dx + b\cdot (1-y) \\
&= \left(y-\frac{1}{e}\right)\cdot b + b\cdot (1-y)\\
& = (1-\frac{1}{e})\cdot b
\end{align}

\end{proof}

This shows that by knowing only the seller's distribution, one can build a price distribution that guarantees a $1-1/e$ fraction of the optimal efficiency regardless of the distribution of the buyer. An immediate corollary of this result is that for every distribution of the buyer there exists at least one price that achieves at least the same fraction of the optimal efficiency (as there always is a price in the support that the expected revenue it guarantees is higher than the expected revenue with respect to the distribution).

\begin{corollary}
For every pair of distributions $\Fs,\Fb$, there exists a deterministic take-it-or-leave-it price that achieves at least $1-1/e$ of the optimal efficiency. This mechanism is dominant-strategy incentive compatible, 
budget balanced and (ex-post) individually rational.
\end{corollary}

The above proposition is non-constructive: to compute such a price one needs to have access to the distribution of the buyer as well. This result was presented informally in the introduction as Theorem 1.

\section{Black-Box Reductions}

The Bilateral Trade problem is one of the most fundamental and well-studied problems in mechanism design. So far, we designed mechanisms that approximate the optimal social welfare for this problem. In this section, we show that the simple Bilateral Trade environment can be used as a building block for the construction of mechanisms in more general environments. We design the mechanisms using black-box reductions: we show that given any $\alpha$ approximation mechanism for Bilateral Trade we can construct a mechanism with an approximation ratio which is some function of this $\alpha$ that maintains all the desired economic properties.

In the rest of this section we will consider a more general model that reallocates a fully divisible good.
In Section \ref{subsec:partership-reduction} we study the $n$-player partnership dissolving setting, where players initially own fractions of the good and have linear utilities. In Section \ref{subsec:reductions:monotone} we study 2-player settings with general monotone valuations, and in Section
\ref{subsec:reductions:convex} we discuss a similar 2-player setting with the restriction to convex valuations (but with an arbitrary initial allocation).


\subsection{Dissolving Partnerships}
\label{subsec:partership-reduction}

In the \emph{partnership dissolving} problem, there are $n$ agents, each agent $i$ owns a share $r_i$ of an asset, and $\sum_{i=1}^n r_i =1$. Each agent $i$ has a value $v_i$ for holding the entire asset, or a value $c\cdot v_i$ for holding a fraction $c\geq 0$ of the asset. Let $r_{max}=\max\{r_1,...,r_n\}$ be the largest share held by an agent.


We show that any approximation mechanism for bilateral trade can be used for constructing a mechanism for partnership dissolving with exactly the same guarantee on the approximation ratio.
This reduction from partnership dissolving to bilateral trade holds for all distributions of the buyers, and irrespectively of the size of the initial shares.
The idea is that each agent sells his share to the other agents
via a second-price auction, where the price taken from the bilateral trade
mechanisms for this setting serves as a reserve price.
We show that for each player, and thus for the whole economy, this mechanism can only improve the efficiency of the bilateral trade mechanism.




\begin{theorem}
\label{thm:reduction-partnership}
Let $M$ be some DSIC, individually rational, and budget balanced mechanism for bilateral trade that achieves an $\alpha$-approximation to the welfare. There is a DSIC, individually-rational and budget-balanced mechanism for partnership dissolving which also achieves an $\alpha$-approximation to the optimal efficiency.
\end{theorem}
\begin{proof}
In the proof we use the already-mentioned fact that any truthful mechanism for bilateral trade simply sets a trade price $p$ that does not depend on the values of the bidders. We develop our mechanism for partnership dissolving in two stages.

\vspace{0.1in} \noindent \textbf{First Stage: A Mechanism where only Bidder $i$ may sell.} We will ``run'' $M$ with bidder $i$ and a hypothetical buyer whose value is distributed according to the distribution of $\max_{k\neq i} v_k$. Let $p$ be the price that $M$ posts. Let $M'$ be the following mechanism:
\begin{itemize}
\item Let $j \in \arg\max_{k\neq i} \{v_k\}$ and let $m_2$ be the second highest value of $v_1,\ldots, v_{i-1},v_{i+1},\ldots, v_n$, that is,
$
    m_2=\max_{k \in N\setminus \{j,i\}} v_{k}.
$
\item Let $p^*=\max\{ p,m_2 \}$. If $v_i\leq p^*$ and $v_j \geq p^*$ then bidder $j$ pays to bidder $i$ the amount of $p^*$ and receives the item. Otherwise bidder $i$ keeps his item and does not get paid.
\end{itemize}

To see that the mechanism is truthful, observe that if a sale is made neither the winning buyer nor bidder $i$ cannot affect the price by changing their bid. A losing buyer $k$ can only turn into a winner by overbidding the winning buyer and paying $p^* \geq v_j\geq v_k$, and therefore cannot gain a positive payoff. If the item is not sold, it is either because the seller's value exceeds $p^*$ (and winning by underbidding induces a payment below $v_i$) or all of the buyers' values are below $p^*$ (and again, overbidding results in a payment higher than $v_j$).

Now for the approximation ratio. Notice that whenever there is a trade in $M$ there is a trade in $M'$ (but the opposite is not true; for example, a trade occurs when $p<v_i\leq m_2$).
$M$ achieves an $\alpha$-approximation to the optimal solution (that may only allocate bidder $i$'s share) which equals $\max\{v_i,\max\{v_{-i}\}\}=\max_{k}\{v_k\}$. $M'$ achieves at least the same expected welfare as $M$, thus it is an $\alpha$-approximation to the optimal welfare (that may only allocate bidder $i$'s share) as well.

The mechanism is individually rational since we always have that $v_i \leq  p^* \leq v_j$. In addition, payments are transferred from one player to another, hence the mechanism is budget balanced.

\vspace{0.1in} \noindent \textbf{Second Stage: The Final Mechanism.} At an arbitrary order, use $M'$ to sell to the other bidders the endowment $r_i$ of each bidder $i$ as a single indivisible item.

We now analyze the approximation ratio. Let $v_{max}=\max_kv_k$ be the highest value. By selling the endowment of bidder $i$ the expected social welfare is at least $\frac{r_i v_{max}}{\alpha}$. Since the valuations of the bidders are linear, after selling all endowments the expected social welfare of at least $\frac{v_{max}}{\alpha}$.

The truthfulness of the mechanism also follows from the linearity of the valuations of the bidders: at every stage they will maximize their payoff from the item independently of the other sales. Therefore, truthfulness follows from the truthfulness of $M'$. Similarly, the mechanism is individually rational and budget balanced.
\end{proof}

As our best approximation for the bilateral-trade problem is $1-\frac{1}{e}$, the reduction guarantees the same approximation ratio for Partnership Dissolving.

\begin{corollary}
There is a DSIC, individually rational, and budget balanced mechanism the Partnership Dissolving problem which is a $1-\frac{1}{e}$ approximation to the optimal effciency.
\end{corollary}

\subsection{2-player Markets with Monotone Valuations}
\label{subsec:reductions:monotone}

We now consider a 2-player market
where a seller initially owns the whole quantity of a divisible good.
The seller and the buyer may have
arbitrary monotone valuations, that is, the only restriction on the valuations is that for fractions $x>y$ we have $v(x)\geq v(y)$.
In this setting, we show that given an $\alpha$-approximation mechanism to the bilateral trade problem, we can design a mechanism with an approximation guarantee of $\frac{\alpha}{1+\alpha}$ that preserves the economic properties of the original mechanism.


\begin{theorem}
Let $A$ be a budget balanced and DSIC mechanism for bilateral trade with one indivisible good that achieves an approximation ratio of $\alpha< 1$. There is a DSIC and budget balanced mechanism $A'$ for bilateral trade with one divisible good that guarantees an approximation ratio of $\frac{\alpha}{\alpha+1}$, as long as the players' valuations are monotonically increasing.
\end{theorem}
\begin{proof}
Let $S$ be the expected contribution of the seller to the welfare-maximizing solution, and let $B$ be the expected contribution of the buyer $B$. Thus, if we let $OPT$ denote the expected optimal welfare, we have that $OPT=B+S$.

Consider some instance and let $s$ be the value of the seller for getting all the item and similarly let $b$ be the value of the buyer for getting all the item. Let $D_s$ be the distribution of $s$ given the seller's distribution and $D_b$ the distribution of $b$. Our mechanism $A'$ will simply run $A$ to obtain a trade price $p$. If $b\geq p\geq s$ the item is re-allocated in full to the buyer and the buyer transfers a payment of $p$ to the seller. Otherwise the seller keeps the full item. Notice that $A'$ inherits from $A$ its truthfulness, budget balance, and individual rationality.

We divide the analysis of the approximation ratio into two cases. In the first case we assume that $S\geq \frac {\alpha} {\alpha+1} OPT$. Observe that in this case the approximation ratio is $\frac{\alpha}{\alpha+1}$: the expected welfare of any mechanism is at at least $E[s]$ since every reallocation cannot decrease the welfare. To conclude this case, observe that clearly we have that by the monotonicity of the valuations, $E[s]\geq S$.

Thus assume that $S<\frac {\alpha} {\alpha+1}OPT$. In other words, $B\geq \frac {OPT} {\alpha+1}$. Notice that if we restrict our attention to instances where the item can only be reallocated in full, then the optimal welfare is at least $E[b]\geq B$ (again, using the monotonicity of the valuations). Thus, using the approximation guarantee of $A$, the expected welfare of $A'$ is at least $\alpha B \geq \frac {\alpha} {\alpha+1} OPT$, as needed.
\end{proof}

\subsection{Trading a Divisible Good with Convex Valuations}
\label{subsec:reductions:convex}

We consider a 2-player market, where each player owns some arbitrary share of a fully divisible good. The initial shares of players $1,2$ are denoted by $r_1,r_2$, where $r_1+r_2=1$.
Each player $i$ has a valuation function $v_i:[0,1]\rightarrow \mathbb R$, and for every $x,y$ we define the marginal valuation $v_i(x|y)=v_i(x+y)-v_i(y)$. We assume that the valuation functions are normalized ($v_i(0)=0$), non decreasing, and have decreasing marginal valuations (i.e., $v_i(\epsilon|x)\geq v_i(\epsilon|y)$ for every $\epsilon>0,y>x$).\footnote{
When $v_i(\cdot)$ is twice differentiable, we simply assume that $v_i^{''}(x) \leq 0$ for every $x$.
}

In this model, we show that given any DSIC and budget balanced $\alpha$-approximation mechanism, we can use it to create an
 $\frac{\alpha}{\alpha+1}$-approximation mechanism for the 2-player market with convex valuations.


\begin{theorem}
Let $A$ be a DSIC and budget balanced mechanism for bilateral trade that achieves an $\alpha < 1$ approximation to the optimal efficiency. There is a DSIC and budget balanced mechanism for the 2-player reallocation problem with arbitrary endowments when the players' valuations are monotone and have decreasing marginals that guarantees an approximation ratio of $\frac{\alpha}{\alpha+1}$.
\end{theorem}
\begin{proof}
Consider the following mechanism:
give the seller for free a fraction $x$ of the item ($x \leq 1$ will be specified later, one of the players that has a share of at least $x$ will be viewed as the seller and we will claim that there is such a player). Denote by $D_s$ the marginal distribution that specifies the value of the seller for receiving an additional $y\in [0,1-x]$ fraction of the item and by $D_b$ the distribution that specifies the value of the buyer for receiving $y\in[0,1-x]$ of the item. Let $i$ be the player that received the item when running the mechanism $A$ on the distributions $D_s$ and $D_b$, and let $p_s\geq 0$ be the payment that the seller receives and $p_b$ the amount that the buyer is charged. Our mechanism give player $i$ a fraction of $1-x$ of the item, additionally allocates the seller $x$, charges the buyer $p_b$ and gives the seller a payment $p_s$. The mechanism is clearly truthful and budget balanced. All that is left is to analyze its approximation ratio.

Let $S$ be the expected value of the seller for receiving all the item, and similarly let $B$ denote the expected value of the buyer for receiving all the item. Clearly, $S+B$ is an upper bound on the expected value of the optimal solution. We will show that for $x=\alpha/(\alpha+1)$ the expected welfare of our mechanism is $x\cdot (S+B)$ and the claim regarding the approximation ratio will follow.

We first bound the contribution of the seller for receiving a fraction $x$ of the item in the first step: Since the valuation of the seller exhibits decreasing marginals the expected contribution is at least $x\cdot S$. We now bound the contribution of the second step. Ideally, we could have calculated the expected optimal value $OPT_2$ of the secondary problem (with the distributions $D_s$ and $D_b$) and claim that the when running the mechanism $A$ the expected welfare is at least $ \alpha\cdot OPT_2$. However, since we do not know how to compute $OPT_2$ exactly we will only give a rough bound to it. Specifically, we claim that $OPT_2\geq (1-x)\cdot B$, which is the welfare we when always allocating the buyer a fraction of $(1-x)$ of the item, since the buyer's valuation function exhibits decreasing marginals.

We get that the overall expected welfare of our mechanism is $x\cdot S+ \alpha (1-x) \cdot B$. Thus, for $x=\frac {\alpha} {\alpha+1}$ the expected welfare is $\frac {\alpha S} {1+\alpha}+\frac {\alpha B} {1+\alpha}$, which gives us an $\frac{\alpha}{\alpha+1}$ approximation ratio.
As there is always a player with at least $\alpha /(\alpha+1)$ fraction of the good (this value is at most 1/2), this player will be the seller.

\end{proof}

\section{Discussion}
\label{sec:discussion}

In this paper, we constructed mechanisms that approximate efficiency in bilateral-trade settings, and showed how such mechanisms can be used as building blocks of mechanisms in more complex scenarios. Our main result states that there always exists a simple, fixed-price, dominant-strategy incentive compatible, strongly budget balanced and ex-post individually rational mechanism that achieves at least $1-1/e$ of the optimal efficiency. This mechanism computes a price as a function of the distribution functions of the seller.

We then show how this implies a $0.39$ approximation for a similar setting with a divisible good and arbitrary monotone valuations, or a 2-player exchange setting for general convex valuations.
We also prove the same $1-1/e$ approximation factor for the $n$ player partnership-dissolving problem. These three approximation results are given as ``black-box" reductions, such that if the $1-1/e$ bound is improved in the future (we know that it cannot be larger than $0.749$ \cite{CKLT16}), these bounds will be immediately improved as well.

A recent working paper \cite{CLMM15} discussed a monopoly pricing problem, where the seller have a limited knowledge on the distribution of the buyer's value and the seller aims to maximize his worst-case profit. While no approximation results are given in \cite{CLMM15}, the structure of the optimal selling mechanism has similarities to our main mechanism (random prices from a distribution with a logarithmic shape). Future work should study whether this new approach can help proving the optimality of such mechanisms in efficiency maximizing scenarios like ours. Another very recent paper is \cite{CS16} that studies equilibria in a game where a buyer chooses a distribution and the seller posts a price after observing this distribution. In equilibrium the seller posts a single deterministic price but the buyer chooses a distribution over the range $[\frac{1}{e},1]$, which again, may indicate that our main mechanism is optimal in some sense.



\subsubsection*{Acknowledgments}

We thank Akaki Mamageishvili for pointing out a mistake in a previous draft.

\bibliographystyle{plain}
\bibliography{partnership}

\appendix

\section{Deterministic 1/2 Approximation: Tightness}

\label{app:subsec:median-is-tight}

\subsubsection*{Proof of Proposition \ref{prop-bilateral-distribution-info}}


We start with the first part. Consider the following distribution $D_s$ of the seller: with probability $\frac 1 2$, $v_s$ gets a value (uniformly at random) in $(0,\epsilon)$. With probability $\frac 1 2$, $v_s$ gets a value (uniformly at random) between $(1,1+\epsilon)$. Observe that the median of $D_s$ is $1$.

Recall that every DSIC, individually rational and budget balanced mechanism in our setting is a posted price mechanism. Now there are two possible cases, depending on the trade price $r$:
\begin{enumerate}
\item $r \leq 1$: let $v_b=\infty$ with probability $1$. The optimal solution always sells the item to the buyer, but the mechanism will sell the item with probability $\frac 1 2$. We get an approximation of $2$ since $v_b >> v_s$.

\item $r > 1$: let $v_b=1-\epsilon$ with probability $1$. The value of the optimal solution is at least $1-\epsilon$ (always sell the item to the buyer). However, the mechanism sells the item only when $v_s \in (0,\epsilon)$, which happens with probability $\frac 1 2$. The approximation is $2$ also in this case.
\end{enumerate}

We now prove the second part. Consider the following distribution $D_b$ of the buyer: with probability $0.99$, $v_b$ gets a value (uniformly at random) in $(1,1+\epsilon)$. With probability $0.01$, $v_b$ gets a value (uniformly at random) between $(100,100+\epsilon)$.

There are several possible cases, depending on the trade price $r$.
\begin{enumerate}
\item $r \leq 1+\epsilon$: let $v_s=1+\epsilon$ with probability $1$. The optimal solution sells the item to the buyer with probability $\frac 1 100$ and the expected welfare is about $2$. The mechanism that post the price $r$ however will never sell the item and will generate an expected welfare of $1+\epsilon$.

\item $1+\epsilon< r \leq 100+\epsilon$: let $v_s=0$ with probability $1$. The expected value of the optimal solution is about $2$ (always sell the item to the buyer). However, the mechanism sells the item only when $v_b \in (100,100+\epsilon)$, so the expected welfare is about $1$. The approximation is $2$ also in this case.

\item $r>100+\epsilon$: let $v_s=0$. The expected value of the optimal solution is $2$, but the mechanism achieves welfare of $0$.
\end{enumerate}

%
%
%
%

\section{Approximating the GFT}

We first show that the GFT cannot be approximated by a budget balanced, individually rational and incentive compatible mechanism to within any constant factor. Notice that this result is proved via examples in which the optimal GFT is a negligible fraction of the social welfare, demonstrating the fact that when the expected GFT is a small fraction of the overall efficiency then approximating it may not be very informative. We also make a straightforward observation that if the fraction of the optimal GFT out of the optimal overall efficiency is large, then an approximation to the optimal efficiency implies a good approximation ratio to the GFT as well (Appendix \ref{app:gft-reduction}).

\subsection{An Impossibility Result for Approximating the GFT}
\label{app:gft-imposs}

Consider a buyer and a seller with values on the support $[0,...,t]$, and
let $\lambda=\frac{1}{1-e^{-t}}$.
Let $F_b(x)=\lambda(1-e^{-x})$ with $f_b(x)=\lambda e^{-x}$ and $F_s(x)=\lambda (e^{x-t}-e^{-t})$ with $f_s(x)=\lambda e^{x-t}$.

\begin{proposition}
For the above distributions, every fixed price mechanism achieves at most $O(1/t)$ approximation to the optimal gain from trade.
\end{proposition}

The optimal gain from trade (divided by the normalization factor of the distributions):
\begin{align*}
\frac{1}{\lambda^2} GFT^\cdot  
    = &  \int_{0}^t \int_{0}^v (v-c)e^{c-t}dc \; e^{-v} dv \\
    = &  \int_{0}^t \left[ e^{c-t}(v-c+1) \right]^v_0 \; e^{-v} dv \\
    = & \int_{0}^t \left( e^{v-t}-e^{-t}(v+1) \right) \; e^{-v} dv \\
    = & e^{-t}t-\int_{0}^t e^{-v-t}(v+1) dv \\
    = & e^{-t}t+\left[ e^{-t-v}(v+2) \right]^t_0 dv \\
    = & e^{-t}t + e^{-2t}(t+2)-2e^{-t}\\
     = & \frac{t-2}{e^{t}}+\frac{t+2}{e^{2t}}
\end{align*}

The gain from trade from posting a price $p$:

\begin{align*}
\frac{1}{\lambda^2}  GFT(p)  
    = & \int_{0}^p \int_{p}^t (v-c)e^{-v} dv \; e^{c-t} dc \\
    = & \int_{0}^p \left[ (c-v-1)e^{-v}\right]^t_p \; e^{c-t} dc \\
    = & \int_{0}^p \big( (c-t-1)e^{-t}-(c-p-1)e^{-p} \big) \; e^{c-t} dc \\
    = & \int_{0}^p (c-t-1)e^{c-2t}-(c-p-1)e^{c-t-p} dc \\
    = & \left[ e^{c-2t}(c-t-2)\right]^p_0 - \left[ e^{c-p-t}(c-p-2) \right]^p_0 \\
    = & e^{p-2t}(p-t-2) - e^{-2t}(-t-2) +2e^{-t} + e^{-p-t}(-p-2) \\
    = & \frac{t+2}{e^{2t}} + \frac{2}{e^{t}} - \frac{p+2}{e^{p+t}} - \frac{e^p(t+2-p)}{e^{2t}} \\
    < & \frac{t+2}{e^{2t}} + \frac{2}{e^{t}}
\end{align*}

Overall, the optimal gain from trade is about $\frac{t-2}{e^{t}}$ which is $O(t)$ more than the gain from trade from any price $p$ which is at most $\frac{2}{e^{t}}$. (Note that the terms that are $O(\frac{t}{e^{2t}})$ are negligible for large $t$'s.)

\subsection{A Simple Reduction from Efficiency to Gains From Trade}
\label{app:gft-reduction}

\begin{observation}
Let $c$ be the fraction of the optimal expected gain-from-trade out of the optimal expected efficiency.
If a mechanism achieves a $x$ fraction of the optimal efficiency, then it achieves at least a $\frac{x+c-1}{c}$ fraction of the optimal gain-from-trade.
\end{observation}


\begin{proof}
We know that $\frac{OPTGFT}{OPT}=c$, where $OPTGFT$ is the optimal expected gain from trade (in the first-best outcome) and $OPT$ is the optimal efficiency.

$OPT=OPTGFT+E[s]$, where $E[s]$ is the expected value of the seller.  Thus, $\frac{E[s]}{OPT} = 1-c$.

Now, the approximation we get for the GFT is ($MECH$ is the expected efficiency of the $x$-approximation mechanism):
\begin{align*}
\frac{MECH-E[s]}{OPT-E[s]}=\frac{\frac{MECH}{OPT}-\frac{E[s]}{    OPT}}{1-\frac{E[s]}{OPT}} \geq \frac{x-(1-c)}{c}
\end{align*}

\end{proof}

For example, for the uniform distribution we have $OPTGFT=1/6, OPT=2/3$ and then $c=1/4$. If we have a mechanism that gains us 90 percent of the welfare, we know it gains at least $(0.9+0.25-1)/0.25=0.6$ of the GFT.

\end{document}